\documentclass{IEEEtran}

\usepackage{graphicx}
\usepackage{amsmath}
\newtheorem{theorem}{Theorem}
\newtheorem{proof}{Proof}
\newtheorem{corollary}{Corollary}
\newtheorem{lemma}{Lemma}

\usepackage{custom_commands}
\usepackage{tikz}
\usetikzlibrary{shapes,arrows,arrows.meta,shapes.geometric,calc}

\usepackage{hyperref}

\title{Power Spectral Density Estimation via Universal Truncated
Order Statistics Filtering}
\author{David Campos Anchieta, John R. Buck}

\begin{document}
\maketitle
\begin{abstract}
Loud transient signals in underwater acoustic data increase the bias and
variance of background noise power spectral density (PSD) estimates based on
sample mean.
Recently, two PSD estimators mitigated the loud transient impact on PSD
estimates by applying order statistics filtering (OSF).
The first, the Schwock and Abadi Welch Percentile, scales a single
rank order statistic (OS) of consecutive periodograms.
The second, the truncated linear order statistics filter, is a weighted
sum of OS up to a chosen rank.
In order to minimize variance, both OSFs must carefully choose the highest rank
that still eliminates the loud transients.
However, in real-time applications in dynamic environments, loud transients
occur at unpredictable rates, requiring dynamic adjustment of the OSF ranks to
keep low bias and variance.
To circumvent the challenges of real-time rank selection, this paper proposes a
convex sum of OSFs across ranks with blending weights that are sequentially
adjusted to favor the lowest variance OSFs over a recent time window.
The performance of the blended sum provably approaches the performance of the
best fixed rank OSF.
Simulations and real data confirm the blended OSFs effectively filter loud
transients out of spectrograms without explicitly choosing a threshold rank.   

\end{abstract}
\section{Introduction}

Many underwater acoustic signal processing applications depend on accurate
background noise power spectral density (PSD) estimates.
Estimating the ambient noise PSD is a critical step in data pre-whitening for
constant false alarm rate (CFAR) detectors \cite{gandhi_analysis_1988,
abraham_underwater_2019}.
Background PSDs also contain information about rain or surface wind for remote
sensing of meteorological events, for example \cite{ma_passive_2005,
mallary_listening_2023, trucco_compounding_2022, trucco_introducing_2023}.
However, underwater acoustic data often include loud transients -- brief,
high-amplitude signals --, which introduce bias and increase the variance of
PSD estimators such as Welch's Overlapped Segment Averaging (WOSA)
\cite{welch_use_1967}.
Estimating the background noise PSD of acoustic environments with loud
transients requires an alternative approach that is more robust to outliers
than the sample mean.

To mitigate the increased bias and variance introduced by loud outliers, order
statistics filters (OSF) exclude data samples that exceed a chosen threshold
rank in the order statistics (OS) when estimating parameters
\cite{rabiner_applications_1975, bovik_generalization_1983, david_order_2003}.
This paper focuses on two OSFs recently applied to PSD estimation.
The first, the Schwock and Abadi Welch Percentile (SAWP) method, replaces the
averaging step in the WOSA with an OS of the periodograms, rescaling each OS so
it is an unbiased estimate of the underlying power for exponential random
variables \cite{schwock_statistical_2021}.
The second, the truncated linear order statistics filter (TLOSF),
computes a weighted average over the lowest-ranked OS
\cite{anchieta_robust_2024, thomson_wt4_1977}.
These weights are chosen to minimize variance while maintaining an unbiased
estimate.
By selecting a threshold rank OS below the maximum, OSFs reduce the impact of
loud transients, while scaling and weighted averaging help maintain unbiased
estimates.



However, the OSFs must carefully choose a fixed threshold rank for each noise
environment — a challenging task in real-time applications due to the
unpredictable nature of loud transients.
Both the frequency of occurrence and the power of the loud transients in
underwater acoustic data are time-varying, requiring real-time selection of the
threshold rank.
Consider two asymptotic extreme cases.
Setting the threshold rank to the minimum ensures maximum robustness to
transients but also maximizes the variance of the OSF.
Conversely, setting the threshold rank to the maximum (for TLOSF) or to the
80th percentile (for SAWP) minimizes variance but increases susceptibility to
bias and mean squared error (MSE) from loud transients.
Ideally, the OSF seeks a rank between these extremes, selecting the highest
threshold rank that still excludes the transients present in the current noise
environment while reducing the estimator variance.
Choosing the threshold rank of an OSF to balance robustness to loud transients
while reducing estimator variance remains a challenging model-order selection
problem.



To circumvent the challenges of rank selection in real-time applications of
OSFs, this paper proposes a performance-weighted sum of OSFs across threshold
ranks.
The blending weights are sequentially adjusted over time to favor the OSFs
performing best over a recent time window.
Several signal processing applications have recently applied this same
universal method for model-order weighting \cite{tucker_performance_2019,
tucker_performance_2025, singer_universal_1999, buck_performanceweighted_2018,
erdim_covariance_2022, erdim_doubly_2022}.
This performance-based blending yields an OSF that is universal over the
selected threshold rank -- an estimator that performs nearly as well as the
best fixed-rank OSF in hindsight -- under the right conditions.
By dynamically adjusting its weights according their recent performance, the
universal OSFs (UOSFs) not only circumvent the threshold rank selection, but
adapt to changes in the occurrence of loud transients over time.
Both UOSFs were effective in smoothing spectrograms of underwater
acoustic data while filtering out loud broadband clicks.

The remainder of this paper is organized as follows.
Section \ref{sec:background} reviews some background on OSFs for PSD estimation.
Section \ref{sec:uosf} describes the universal method by model-order weighting
and its application on a UOSFs.
Section \ref{sec:results} presents simulation results comparing the universal
SAWP and universal TLOSF with their fixed-rank counterparts, and real data
comparing the UOSFs against WOSA in spectrograms of hydrophone recordings.
Finally, Section \ref{sec:conclusion} summarizes the results and draws
conclusions.

\section{Background}
\label{sec:background}

This section is divided into four parts that present background material while
establishing the mathematical notation of this paper.
The first part gives a brief description of the WOSA spectral estimator.
The second part describes the SAWP estimator as an alternative to the WOSA.
The third part presents the TLOSF as an alternative OSF.
The last part introduces a generalized mathematical representation for the OSFs,
which will be convenient when developing a universal version of both.

\subsection{Welch's Overlapped Segment Averaging}

The WOSA is a classic spectral estimator that averages periodograms of
overlapping segments of data into a background noise PSD estimate
\cite{welch_use_1967}.
Each periodogram is the squared magnitude of a tapered segment of the time
domain signal \cite{kay_modern_1988}:
\begin{equation}
	\pgram{n} = \frac{1}{L} \left| \sum_{l=0}^{L-1}
	w[l]x[n D+l] e^{-j2\pi f_0 l} \right|^2,
	\label{eq:welch_periodogram}
\end{equation}
where $x[l]$ is the signal in the time domain, $w[l]$ is the taper window
function, $L$ is the number of time samples producing each periodogram, $D$ is
the shift between segments in number of samples, $0\leq n \leq N-1$ with $N$
being the total number of periodograms,  and $0\leq f_0 \leq 0.5$ is the
frequency bin normalized by the Nyquist rate.
Each short-time WOSA PSD estimate is an average of a number $R$ of consecutive
periodograms:
\begin{equation}
\begin{gathered}
	\pwosa{t}= \frac{1}{R} \sum_{n=0}^{R-1} \pgram{tQ + n},
	\label{eq:wosa}
\end{gathered}
\end{equation}
where $R\leq N$ and $0<t<T-1$, with $T$ being the total number of short-time
PSD estimates, and ($R-Q$) is the overlap between consecutive PSD estimates in
number of periodograms.
Each PSD estimate incorporates $[(R-1)D+L]/f_s$ seconds of the time domain
signal, where $f_s$ is the sampling rate.
The overlaps between periodograms $D$ and between PSD estimates $Q$ can be
bigger or smaller according to the desired time resolution for the PSD
estimates.
In order for the OSF to filter loud transients, each PSD estimate must
incorporate a segment of the input signal that is longer than the longest loud
transient.





\subsection{The SAWP}

Loud transients -- occasional loud signals of short duration -- in underwater
acoustic data introduce bias and increase variance in estimators based on the
sample mean, such as the WOSA, thereby limiting their effectiveness.
An estimator derived from the sample median of the periodograms would mitigate
the bias caused by the loud transients \cite{beaton_fitting_1974,
rabiner_applications_1975, allen_findchirp_2012}.
Seeking a more efficient estimator, the SAWP estimator extends the range of
rank-based OS options beyond the sample median \cite{schwock_statistical_2021}.


OSFs like the SAWP start by sorting the periodograms by magnitude from the
quietest to the loudest, such as:
\begin{equation*}
	\porder{t} \in \pgramset, \\
\end{equation*}
such that
\begin{equation*}
	\porder{t}  \leq \porder[r+1]{t}, r = 1,2,\dots,R,
\end{equation*}
where $\porder{t}$ is the $r$-th rank OS of the set of periodograms yielding a
short-time PSD estimate.
The $\tr$-th rank OS is also the $100 (\tr/R)$ percentile.
The magnitude-ordering of the periodograms concentrates the occasional loud
transients on the highest ranks, leaving clean background samples on the lowest
ranks.
However, despite being free of the loud transient bias, the low-rank OS of the
periodograms are still biased estimates of the background noise PSD.
The $\tr$-th rank SAWP estimator corrects the bias of the OS of the
periodograms normalizing it by a scaling constant
\cite{schwock_statistical_2021, anchieta_robust_2024}:
\begin{equation}
	\psawp[\tr]{t} = \porder[\tr]{t}/\alpha_{r_0}.
	\label{eq:sawp}
\end{equation}
The scaling constant $\alpha_{r_0}$ is the mean of the $\tr$-th rank OS of
a sample of $R$ independent and identically distributed (IID) realizations of a
standard (unit mean) exponential random variable, which is given by the sum of
fractions \cite{anchieta_robust_2023, david_order_2003,
schwock_statistical_2021, gandhi_analysis_1988}:
\begin{equation}
	\alpha_{r_0}  = \sum_{k=R-\tr+1}^{R} \frac{1}{k}.
	\label{eq:indep_increments_mean}
\end{equation}
The normalization in (\ref{eq:sawp}) assumes that the time-domain signal $x[l]$
is Gaussian-distributed and stationary.
From that assumption, it follows that the short-time periodograms $\pgram t $
can be modeled as a standard exponential random variable scaled by the true PSD
\cite{schwock_statistical_2021, anchieta_robust_2024}.
Therefore, scaling the OS of the periodograms by $1/\alpha_{r_0}$ results in an
unbiased estimator of the background noise PSD.



\subsection{The TLOSF}


If outliers are absent, the SAWP estimator is unbiased regardless of the chosen
threshold rank $\tr$.
If outliers are present, the SAWP estimator with threshold rank $\tr$ will
remain unbiased as long as the threshold rank is low enough to exclude the loud
transients.
However, the SAWP discards the information contained in the lower rank OS
$r<\tr$ even though those data are also free from the loud transient bias.
The TLOSF incorporates the information of those lower rank OS to further reduce
the variance of the PSD estimate.

The TLOSF is a weighted average of all the OS of the periodograms up to the
threshold rank:
\begin{equation}
	\ptlosf{t} = \sum_{r=1}^{\tr} w_r \porder[r]{t}.
	\label{eq:tlosf_def}
\end{equation}
The weights $w_r$ result from a least squares optimization that
minimizes the variance of the resulting estimator while keeping it
unbiased \cite{anchieta_robust_2024, thomson_spectrum_1977,
lloyd_leastsquares_1952}:
\begin{equation}
w_r =
\begin{cases}
1/\tr, &r<\tr;\\
(R-\tr+1)/\tr, &r=\tr.\\
\end{cases}
\end{equation}
When both SAWP and TLOSF use rank $\tr$ and there are no outliers at rank
$\tr$ or below, TLOSF reduces the variance by about 0.5 dB compared to
SAWP.
Conversely, when choosing SAWP and TLOSF to have the same variance, TLOSF will
have a lower threshold rank, making it more robust to loud transients
\cite{anchieta_robust_2024}.


\subsection{A Generalized Representation of the Order Statistics Filters}
\label{sec:generalized_osf}
For the remainder of this paper, it will be convenient to represent both SAWP
and TLOSF using a common notation of an inner product between a weight vector
and a vector of OS of the periodograms:
\begin{equation}
	\posf{t} = \osw \pordervec{t},
	\label{eq:posf}
\end{equation}
where
\begin{equation*}
	\pordervec{t} = [\porder[1]{t}, \dots, \porder[R]{t}]^T
\end{equation*}
is a vector of OS of consecutive periodograms sorted from smallest to largest.
For SAWP, the weight vector has a single nonzero entry:
\begin{equation}
	\osw = (1/\alpha_{r_0})[\vc{0}_{\tr-1}, 1, \vc{0}_{R-\tr}],
	\label{eq:sawp_weights}
\end{equation}
where $\vc{0}_{N}$ is a vector of $N$ zeros.
For TLOSF, the weight vector is nonzero for all entries on or below the
threshold rank:
\begin{equation}
	\osw = (1/\tr) [\vc{1}_{\tr-1}, (R-\tr+1), \vc{0}_{R-\tr}],
	\label{eq:tlosf_weights}
\end{equation}
where $\vc{1}_N$ is a $N$-dimensional vector of ones.
This common notation provides a convenient shared approach for the development
of a universal version of both OSFs.

\section{The Universal Order Statistics Filter}
\label{sec:uosf}

While both TLOSF and SAWP mitigate the bias and variance caused by loud
transients in data, the selection of the threshold rank is a challenging
model-order selection problem in practical real-time applications.
Choosing a lower threshold rank OS ensures robustness against occasional loud
transients, but generally increases the variance of the OSF.
Fig. \ref{fig:varosf} corroborates this by plotting the normalized variance of
TLOSF and SAWP versus threshold rank for $R=100$ as provided in
\cite{anchieta_robust_2024}.
For TLOSF, the variance is minimized at the maximum rank, while for SAWP it is
minimized around the 80th percentile.
Additionally, the power and occurrence of real world loud transients are
time-varying, requiring the OSF to adjust its threshold rank in real-time to
adapt to the changes in the environment noise.

\begin{figure}
	\centering
	\includegraphics[width=.9\linewidth]{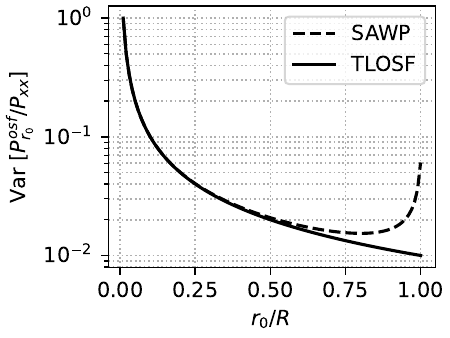}
	\caption{Normalized variance of the fixed-rank OSF as a function of the
	threshold rank $\tr$ for $R=100$.
	SAWP (dashed line) achieves its minimum variance with $\tr \approx 0.8 R$
	while TLOSF (solid line) achieves minimum variance with $\tr = R$.
	A lower threshold rank $\tr$ increases robustness against loud
	transients, but generally also increases variance.}
	\label{fig:varosf}
\end{figure}

The UOSF solves the rank selection problem of the OSFs by blending the
best performing fixed-rank OSFs according to their short-time unnormalized
sample second moment using a universal method for model-order weighting
\cite{singer_universal_1999}.
Instead of choosing a single threshold rank as in (\ref{eq:sawp_weights}) and
(\ref{eq:tlosf_weights}), the weights of a UOSF estimator are a linear
combination of fixed-rank weight vectors:
\begin{equation}
	\wu[t] = \sum_{r=1}^R \univmu{r}{t} \osw[r]
	\label{eq:univ_weights}
\end{equation}
where $\univmu{r}{t}$ is the blending weight for the $r$-th rank weight
vector at the time iteration $t$.
The softmax of the scaled loss function over the threshold ranks gives the
blending weights at each time iteration:
\begin{equation}
	\univmu{r}{t} = \frac{\exp \left( - \tfrac{1}{2c} \lossf[r]{t-1} \right)}
	{\sum_{k=1}^R \exp \left( - \tfrac{1}{2c} \lossf[k]{t-1} \right)},
	\label{eq:univ_coeffs}
\end{equation}
where $\lossf[r]{t}$ is the loss of the $r$-th rank OSF at the iteration $t$,
and $c$ is a hyperparameter that controls the blending weight's sensitivity to
the loss function.
Note that the blending weights are monotonically decreasing as a function of
the loss, meaning that OSF ranks with lower loss are assigned higher blending
weights.
The loss function is the sum of squares of the estimator over a fixed number of
previous iterations:
\begin{equation}
	\lf{t}{\vc{w},\hatpvec} = \sum_{t'=t-\tau+1}^t (\vc{w}[t'] \pordervec{t'})^2
	\label{eq:loss_function}
\end{equation}
where $\tau$ is the window size of the loss function in time iterations.
This loss definition holds for the fixed-rank weight vectors $\osw[r]$
as well as the universal weight vectors $\wu$, observing that the
fixed-rank weights are not time-dependent, i.e.  $\osw[\tr][t] =
\osw[\tr] \; \forall \; t$.
The inner product between the blended weight vector in (\ref{eq:univ_weights})
and the vector of OS of the short-time periodograms result in a UOSF spectral
estimator:
\begin{equation}
	\puosf {t} = \wu[t] \pordervec{t}.
	\label{eq:puosf}
\end{equation}
The block diagram in Figure \ref{fig:block_diagram} illustrates how the UOSF
estimator is generated.
The denominator in (\ref{eq:univ_coeffs}) guarantees that the blending weights
add up to one, $\sum_{r=1}^R \univmu{r}{t} = 1$, making $\puosf {t}$ a
convex combination of unbiased estimators, which is itself unbiased in pristine
data\footnote{Note that combining (\ref{eq:posf}),  (\ref{eq:univ_weights}),
and (\ref{eq:puosf}) results in $\puosf {t} = \sum_{r=1}^R \univmu{r}{t}
\posf[r]{t} $.}.

Choosing the unnormalized sample second moment as the loss function in
(\ref{eq:loss_function}) indirectly minimizes the MSE of the UOSF by minimizing
its variance.
Since the blending weights in (\ref{eq:univ_coeffs}) monotonically decrease as
a function of the loss, the ideal loss function reflects an optimality
criterion the UOSF intends to minimize.
While the MSE is a natural optimality criterion for parameter estimators such
as the SAWP and TLOSF, its dependence on the true background noise PSD makes it
unavailable in practical applications \cite{kay_fundamentals_1993}.
The sample second moment, however, can be determined in real-time applications
and, as Lemma \ref{lem:variance} states, is itself an upper bound on the
variance of the OSF.
Therefore, minimizing the sample second moment minimizes the variance of the
OSF.
For estimators that are unbiased or nearly so, i.e.  $\expectation{\vc{w}
\hatpvec} = P_{xx}$, where $P_{xx}$ is the true background noise PSD,
minimizing the variance essentially minimizes the MSE.

\begin{lemma}
Let $\expectation{\vc{w} \hatpvec} = P_{xx}$. Then:
\label{lem:variance}
\begin{equation}
\begin{aligned}
	\var{\vc{w} \hatpvec}
		&= \frac{1}{\tau} \lf{t}{\vc{w}, \hatpvec}
		- \left(\frac{1}{\tau} \sum_{t'=t-\tau+1}^t \vc{w}[t'] \hatpvec [t']
		\right)^2 \\
		&\approx \frac{1}{\tau} \lf{t}{\vc{w}, \hatpvec}
		- P_{xx}^2
\end{aligned}
\label{eq:variance}
\end{equation}
\end{lemma}

Under the right conditions, the difference in performance between the UOSF
and the best fixed-rank OSF, as measured by the variance, vanishes
asymptotically with growing $\tau$.
Theorem \ref{thm:regret} sets the upper bound on the \emph{per sample} loss of
the UOSF as the per sample loss of the best fixed-rank OSF plus some residual.
Defining the \emph{regret} as the difference between the variance of the UOSF
and the variance of the best fixed-rank OSF, it follows from
(\ref{eq:variance}) that:
\begin{equation}
\begin{aligned}
	\varepsilon_t \left( \wu, \hatpvec \right) &\stackrel{\Delta}{=}
		\var{\wu \hatpvec} - \min_r \var{\osw[r] \hatpvec}\\
		&= \tfrac{1}{\tau}\lf{t}{\wu, \hatpvec}
		- \tfrac{1}{\tau} \min_r \lossf[r]{t}.
\end{aligned}
\end{equation}
The corollary of Theorem \ref{thm:regret} concludes that, as the loss function
window size $\tau$ grows, the regret approaches zero.

\begin{theorem}
\label{thm:regret}
Let $\osw[r]\pordervec{t} \leq A \; \forall \; t$, and $t \geq \tau$.
Then:
\begin{equation}
\tfrac{1}{\tau}\lf{t}{\wu, \hatpvec}
	\leq \tfrac{1}{\tau} \min_r \lossf[r]{t} + \tfrac{2A}{\tau}\ln R.
	\label{eq:lossbound}
\end{equation}
\end{theorem}

\begin{proof}
See appendix.
\end{proof}


\begin{corollary}
From (\ref{eq:lossbound}), it follows that
\begin{equation*}
	\lim_{\tau \to \infty} \varepsilon_t \left( \vc{w}[t], \hatpvec \right) = 0.
\end{equation*}
\end{corollary}

This universal framework applies equally to any unbiased OSF, unifying them
under a common structure based on linear combinations of fixed weight vectors
from Sec. \ref{sec:generalized_osf}.
The specific form of the UOSF depends on the choice of these competing
weight vectors.
For the results on this paper, the weight vectors competing over rank in a
UOSF are either of TLOSF form or SAWP form, not a combination of both.
If the competing weight vectors are in the SAWP form as in
(\ref{eq:sawp_weights}), the resulting OSF is a universal SAWP (USAWP);
if the competing weight vectors are in the TLOSF form as in
(\ref{eq:tlosf_weights}), the resulting OSF is a universal TLOSF (UTLOSF).
Even though this paper focuses on universal versions of the TLOSF and SAWP,
the competing weight vectors in (\ref{eq:univ_weights}) can be of any subset of
unbiased OSFs $\mathcal{W} \subseteq \{\osw[]|
\mathrm{E}\{\osw[]\hatpvec\}=\psd\}$, with $ \wu[t] = \sum_{\osw[k] \in
\mathcal{W}} \univmu{k}{t} \osw[k]$.

\begin{figure*}
	\centering
	\includegraphics[width=.8\linewidth]{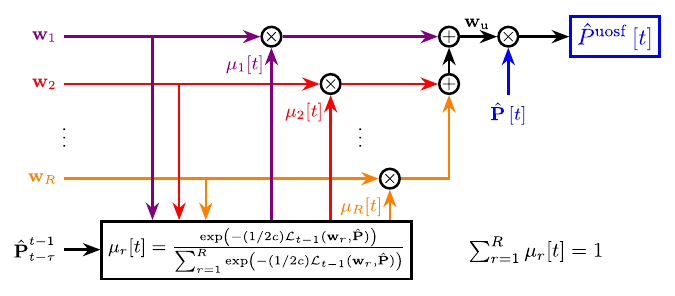}
	\caption{Block diagram describing the flow of information resulting in an
	UOSF.}
	\label{fig:block_diagram}
\end{figure*}

\section{Results}
\label{sec:results}

This section compares the performance of the USAWP, UTLOSF, and their
fixed-rank counterparts in two parts.
The first set of results compare the performance of the UOSFs and the
fixed-rank OSFs in computer simulations with synthetic data that mimics
periodograms of Gaussian distributed signals.
The second set of results compare the UOSFs to the WOSA on spectrograms of an
actual hydrophone recording.

The simulated data is created by a mixture containing two exponential random
variables with distinct means to model a background power of $\lambda$ with
occasional transients $K$ times louder that happen with probability $\rho[t]$:
\begin{equation}
\begin{aligned}
	f_{\tilde{\vc{P}}[t]} (\vc{x}) = (1-\rho[t]) \tfrac{1}{\lambda} &\exp (-\vc{x}/\lambda) +\\
	\rho[t] \tfrac{1}{K \lambda} &\exp(-\vc{x}/K \lambda).
\end{aligned}
\end{equation}
The OSFs estimate the target mean $\lambda$ from sorted realizations of the
simulation data $\hatpvec[t] = \sort (\tilde{\vc{P}}[t])$.
For the following results, the simulation parameters are: $\lambda = 1$,
$K=200$, $R=20$, $c = 1$, $\tau=250$, total iterations $T =
3000$.
The probability of outliers $\rho[t]$ starts at zero and increases by 0.02 (two
percent points) after every 500 iterations.

In the Monte Carlo trials, the UOSFs nearly match the variance and
MSE of their best fixed-rank counterpart OSF.
Figure \ref{fig:varmse1} plots the variance (left column) and MSE (right
column) for both the fixed-rank SAWP OSFs (top row, colored squares) and
fixed-rank TLOSF (bottom row, colored triangles) compared to their universal
counterparts (black symbols) as a function of the probability of outliers.
The color indicates the quantiles of the fixed-rank OSFs from lowest (purple)
to highest (yellow).
Each plotted point is an average of 3000 Monte Carlo trials at the iterations
499, 999, 1499, 1999, 2499, and 2999, which are the last iterations before the
probability of outliers increases.
As the probability of outliers increases from zero, the variances and MSEs of
the high-quantile fixed-rank OSFs increase by several orders of magnitude.
As the probability of outliers continues increasing, fixed-rank OSFs of modest
quantiles also suffer from outliers, increasing their variances and MSEs.
No single fixed-rank OSF remained as the minimum variance estimator over all
the probability of outliers.
The UOSFs preserve a variance that nearly matches the best fixed-rank
estimator by blending the fixed-rank estimators according to
(\ref{eq:univ_weights}) and (\ref{eq:univ_coeffs}).
Across the probability of outliers, the MSE of the USAWP is less than 0.6 dB
higher than the MSE of the best SAWP, while MSE of the UTLOSF is less than 0.8
dB higher than the MSE of the best TLOSF.

\begin{figure}
	\centering
	\includegraphics[width=\linewidth]{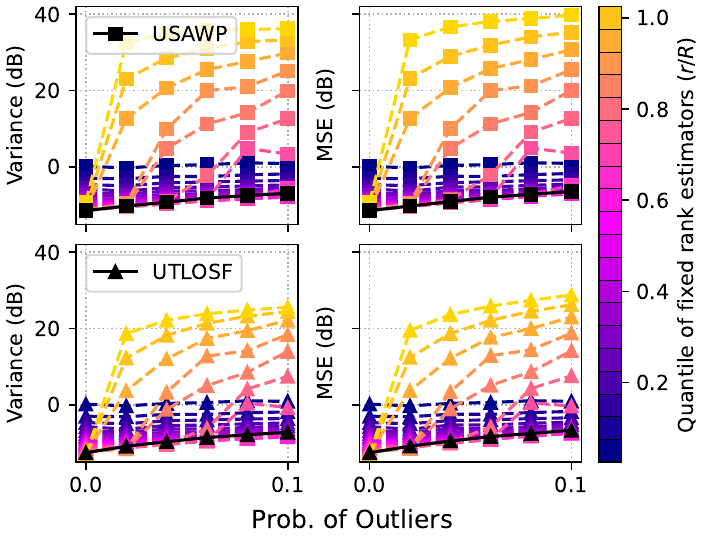}
	\caption{Variance (left panels) and MSE (right panels) of USAWP (top
	panels, black squares) and UTLOSF  (bottom panels, black triangles)
	compared to their fixed-rank counterparts (colored symbols) as a function
	of the probability of outliers in data.
	Color indicates the quantiles of the fixed-rank OSFs from lowest (purple)
	to highest (yellow).
	As the probability of outliers increases, both variance and MSE of the
	fixed-rank estimators increase by several orders of magnitude.
	In contrast, the UOSFs (black symbols) preserve a lower variance than
	any of their fixed-rank counterparts, leading to an overall lower MSE.}
	\label{fig:varmse1}
\end{figure}

The UTLOSF achieves a slight MSE reduction when compared to USAWP mostly due to
a reduction in variance.
Figure \ref{fig:varmse2} replots all the black symbols from Fig.
\ref{fig:varmse1} on the same axis to compare the variance and MSE of the USAWP
and UTLOSF estimators directly.
The UTLOSF performs slightly better than USAWP, having a 1.2 dB MSE
reduction in pristine data, and 0.3 dB MSE reduction in
data with 10\% of outliers.
As the probability of outliers increases, the difference between MSE and
variance increases for both estimators, indicating an increasing the
contribution to the MSE from the bias squared term.
The overall low difference (of up to 0.5 dB) between MSE and variance reveals
that the variance, not bias squared, dominates the MSE of the estimators,
especially in data with fewer outliers (smaller $\rho[t]$).

\begin{figure}
	\centering
	\includegraphics[width=0.6\linewidth]{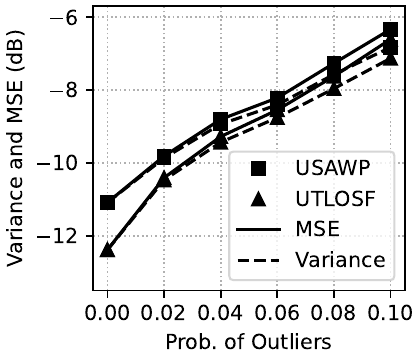}
	\caption{Variance (dashed lines) and MSE (solid lines) of the universal
	OSFs, as a function of the probability of outliers.
	As the occurrence of outliers increases, the gap between variance and MSE
	increases for both estimators, indicating a growing bias squared.
	However, the overall low difference between MSE and variance indicates that
	the MSE of both UTLOSF and USAWP increases mainly due to increasing
	variance.
	The UTLOSF (triangles) maintains a slightly lower variance than USAWP
	(squares), which resulted in a slightly lower MSE as well.}
	\label{fig:varmse2}
\end{figure}

The UOSFs quickly respond to the increasing occurrence of outliers by
reducing the blending weights of fixed-rank estimators whose sample second
moments were increased by outliers.
Consequently, this increases the blend weights of the lower-rank estimators
with lower second moments.
Figure \ref{fig:blend_weights} plots the blending weights of the
UTLOSF (top panel) and USAWP (bottom panel) as a function of
the time and quantile averaged over the Monte Carlo trials.
The color at each point represents the blend weight for that fixed-rank
estimator at that time.
In the first 500 iterations, when there are no outliers, the UTLOSF sets the
highest blending weight to the highest threshold rank, making the UTLOSF equal
to the sample mean, which is known to be the efficient estimator in the absence
of outliers.
Meanwhile, the USAWP sets the highest blending weight to the threshold rank
with quantile $r/R \approx 0.8$.
As shown in Fig. \ref{fig:varosf}, those are the threshold ranks that
minimize the variance of each OSF in pristine data.
As the probability of outliers increases after every 500 iterations, both
UOSFs quickly adjust their weights to increase the blending weights of
lower threshold ranks, while reducing the blending weights of higher ranks that
are now vulnerable to outliers.
This quick adaptation to the changes in data guarantees that the UOSFs'
performances are comparable to their best fixed-rank counterparts.

\begin{figure}
	\centering
	\includegraphics[width=0.98\linewidth]{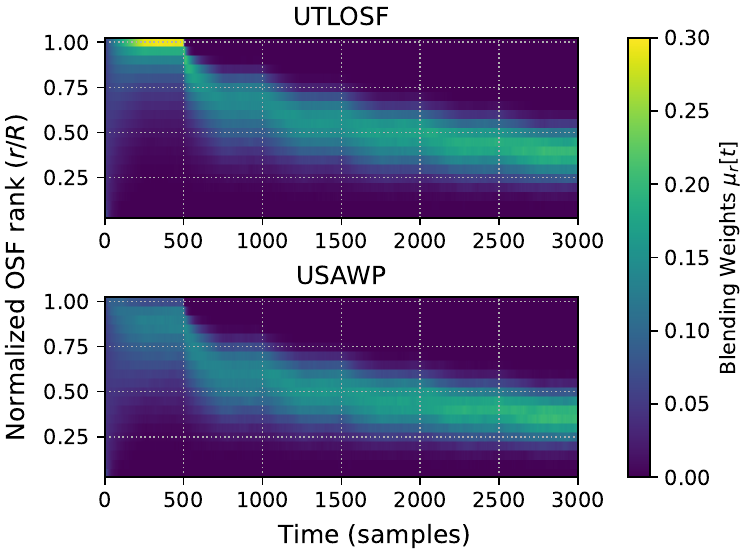}
	\caption{Color plot of the blending weights of the UOSFs as a
	function of quantile and time averaged over Monte Carlo trials, showing
	adaptation to the increasing probability of outliers.
	The probability of outliers in the simulation data increases by two percent
	points after every 500 time iterations.
	As the occurrence of outliers increases, the UOSFs reduce the
	blend weights of higher rank estimators, which are more prone to outlier bias,
	and increase the blend weights for the lower OSFs.}
	\label{fig:blend_weights}
\end{figure}

To evaluate the performance of the proposed UOSFs in real-world conditions, we
applied them to underwater acoustic data containing impulsive noise.
Figure \ref{fig:dock_data} compares the WOSA and both UOSFs when applied as
sliding windows to the spectrogram of a hydrophone recording.
The top-left panel displays the spectrogram of a recording obtained from a
hydrophone at the UMass Dartmouth School of Marine Science and
Technology in New Bedford, MA.
The HTI-96 hydrophone was about 2 m deep in water that varies between 3-4 m
deep with the tide.
The particular segment was recorded with a Zoom H5 sampling at 44.1 kHz on
August 21st, 2021, at 4 p.m..
Each vertical line of the spectrogram is a periodogram obtained with $L=512$
point Hanning window and $D=128$ point step between samples.
The bright vertical lines in the spectrogram are broadband clicks that are up
to 30 dB louder than the background noise, making them an example of loud
transients.

Both UOSFs smooth the spectrograms of underwater acoustic data while also
filtering occasional loud clicks in the recording.
The bottom-left panel plots the WOSA spectrogram with $R=100$ and $Q=1$.
All the spectrograms have triangles pointing at the same times and frequencies
of some occurrences of loud clicks in the recording, tracking the effect of the
different PSD estimators on those loud transients.
While the WOSA spectrogram is smoother than the raw unaveraged spectrogram (top
left panel), the WOSA spectrogram still has the very bright vertical lines due
to the loud broadband clicks in the recording.
The time-averaging of the WOSA spectrogram increases the width of these loud
outliers in time.
In contrast, the panels on the right side plot the UOSF spectrograms computed
with $\tau=500$ and $c=10^5$.
Both USAWP and UTLOSF result in smooth spectrograms, but without the bright
lines, meaning that they effectively filter the impulsive noise in the input
data.
The UOSF results are undistinguishable from each other because, as shown in
Fig. \ref{fig:varmse2}, the difference in performance between USAWP and UTLOSF
is small enough to not result in perceptual difference in this visualization.

\begin{figure*}
	\centering
	\includegraphics[width=0.82\linewidth]{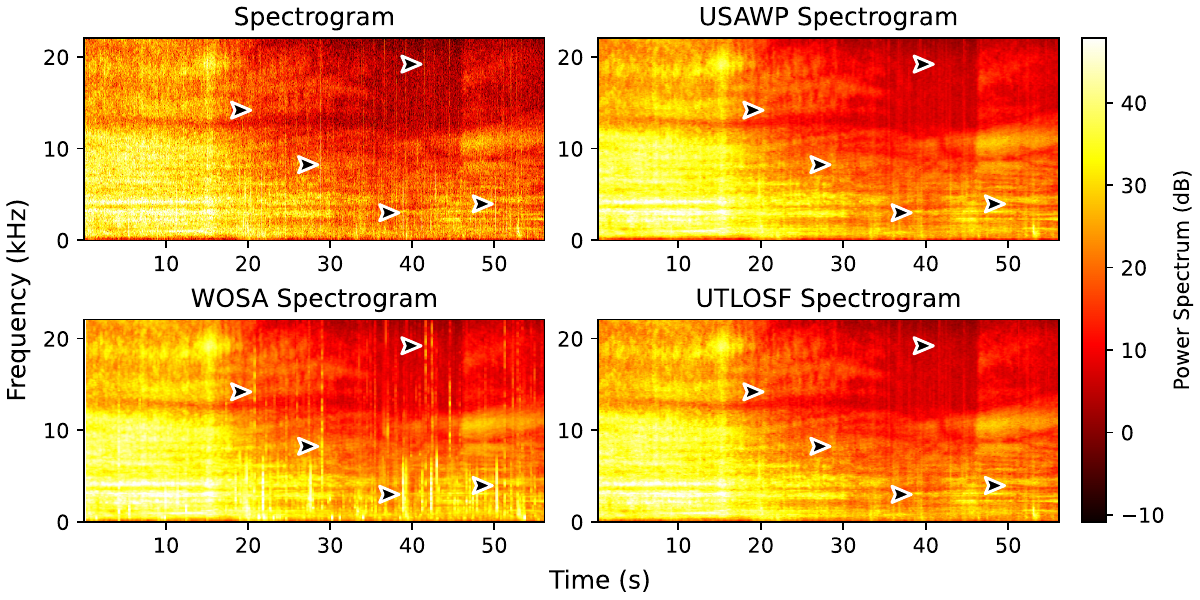}
	\caption{Comparison of the raw (top-left), WOSA (bottom-left), USAWP
	(top-right), and UTLOSF (bottom-right) spectrogram estimates.
	The WOSA spectrogram exhibits bright vertical lines caused by broadband
	clicks in the hydrophone recordings.
	In contrast, both UOSF spectrograms (USAWP and UTLOSF) effectively suppress
	these transients, resulting in cleaner and smoother spectrograms.
	The triangles point to the time and frequency of some occurrences of loud
	broadband clicks in the recording.}
	\label{fig:dock_data}
\end{figure*}

\section{Conclusion}
\label{sec:conclusion}

This paper proposes universal SAWP and TLOSF PSD estimators that are adaptive
over the chosen threshold rank, making them more suitable for real-time
background noise PSD estimation in environments with occasional loud
transients.
The proposed UOSFs blend the fixed-rank OSFs according to their short-time
performance as measured by the unnormalized sample second moment.
By leveraging the sample second moment as a performance criterion, the UOSF
dynamically suppresses the contribution of fixed-rank OSFs impacted by loud
transients, reducing the variance and MSE.
The universality of the method provides provable guarantee that the variance of
the UOSF processing data online asymptotically converges to the variance of the
best fixed-rank OSF that one would choose in hindsight if processing the data
offline in batch mode.
Moreover, these performance guarantees hold for any individual data sequence
bounded in amplitude.
Monte Carlo simulations using synthetic data confirm that both the USAWP
and UTLOSF estimators maintain MSEs comparable to the best fixed-rank OSFs,
while simultaneously mitigating the bias introduced by loud transients.
When applied to real hydrophone recordings containing impulsive noise, the
UOSFs effectively suppressed loud broadband clicks in the recordings, yielding
smooth spectrograms of the background noise.
The UOSF would be beneficial in the underwater acoustic applications that
depend on accurate and precise estimators of the background noise PSD, such as
CFAR detection, and meteorological remote sensing.

\section*{Acknowledgements}

This project is supported by the ONR Code 321US on project N00014-23-1-2133.
The authors would like to thank CJ Berg for providing the hydrophone data.

\section*{Author Declarations}
\subsection*{Conflict of Interest}

The authors have no conflict of interest to disclose.


\appendix[Proof of Theorem \ref{thm:regret}]
To establish performance bounds of the blended UOSFs relative to their
best fixed-rank OSFs, this proof compares both to a hypothetical
uniformly weighted OSF.
The uniformly weighted OSF is only hypothetical because it is not meant to be
implemented, but serves as a common reference point to compare the performances
of both OSFs in order to establish the desired inequality.
As a mathematical artifice to compare the performances of the UOSFs to their
respective OSFs, this proof defines functions that turn the loss of the OSFs
into pseudo-probabilities of assignment, however the proof does not require a
specific probability distribution to hold.
The bounds established hold for any individual data sequence satisfying the
assumptions stated in Theorem \ref{thm:regret}.

Considering the loss function of an arbitrary OSF with weight vector
$\vc{w}$ at the time iteration $\tau$:
\begin{equation*}
\lossf[]{\tau} \stackrel{\Delta}{=} \sum_{t=1}^\tau (\vc{w}^T \pordervec{t})^2,
\end{equation*}
and the pseudo-probability of a r-th rank OSF as function of the loss:
\begin{equation*}
\mathrm{p}_r (\hatpvec^\tau) \stackrel{\Delta}{=} B \exp \left(-(1/2c) \lossf[r]{\tau}\right),
\end{equation*}
where $\hatpvec^\tau$ is the set of consecutive samples of ordered periodograms
$\{\pordervec{1}, \dots, \pordervec{\tau}, \}$.
From the pseudo-probability of the fixed-rank OSFs follows the pseudo-probability of
the uniformly weighted OSF:
\begin{equation*}
\mathrm{p}_\mathrm{avg} (\hatpvec^\tau) = \frac{1}{R} \sum_{r=1}^R \mathrm{p}_r (\hatpvec^\tau).
\end{equation*}
From the previous two definitions and the nonnegativity of $\mathrm{p}_r
(\hatpvec^\tau)$, we can establish the first inequality:
\begin{equation}
\mathrm{p}_\mathrm{avg} (\hatpvec^\tau) \geq \frac{1}{R}  \max_r \mathrm{p}_r
(\hatpvec^\tau),
\label{eq:pavg_inequality}
\end{equation}
In this context, the maximum pseudo-probability $\mathrm{p}_r (\hatpvec^\tau)$
across the threshold ranks $r$ is the pseudo-probability of the best fixed-rank
OSF, i.e., the fixed rank OSF with the minimum loss across the threshold ranks.

For the next inequalities, it will be useful to express the r-th rank OSF
pseudo-probability as a product of a sequence of conditional pseudo-probabilities.
The pseudo-probability of the r-th rank OSF at time $\tau$ conditioned on the previous
data samples is:
\begin{equation*}
\begin{aligned}
\mathrm{p}_r (\pordervec{\tau} | \hatpvec^{\tau-1})
	&= \frac{\mathrm{p}_r (\hatpvec^\tau)}{\mathrm{p}_r (\hatpvec^{\tau-1})}\\
&= \exp \left(-\tfrac{1}{2c} (\osw[r] \pordervec{\tau})^2\right)\\
	&=f_\tau(\osw[r])
\end{aligned}
\end{equation*}
where
\begin{equation*}
f_t(\vc{w}) \stackrel{\Delta}{=} \exp (-\tfrac{1}{2c} (\vc{w}^T \pordervec{t})^2 ).
\end{equation*}
Successive conditioning expresses the pseudo-probability of the r-th rank OSF
as a product of conditional pseudo-probabilities:
\begin{equation*}
\begin{aligned}
\mathrm{p}_r (\hatpvec^\tau)
	&= \prod_{t=1}^\tau
\mathrm{p}_r (\pordervec{t} | \hatpvec^{t-1})\\
&= B \prod_{t=1}^\tau f_t(\vc{w}),
\end{aligned}
\end{equation*}
with $\mathrm{p}_r (\pordervec{1} | \hatpvec^{0}) = \mathrm{p}_r
(\hatpvec^{1})$.

These same principles apply when determining the uniformly-weighted OSF
pseudo-probability conditioned on the previous samples of data:
\begin{equation*}
\begin{aligned}
\mathrm{p}_\mathrm{avg} (\pordervec{\tau} | \hatpvec^{\tau-1}) &=
	\frac{\sum_{r=1}^R \mathrm{p}_r (\hatpvec^\tau) }
{\sum_{k=1}^R \mathrm{p}_k (\hatpvec^{\tau-1}) }\\
&=\frac{\sum_{r=1}^R \mathrm{p}_r (\pordervec{\tau} | \hatpvec^{\tau-1})
	\mathrm{p}_r (\hatpvec^{\tau-1}) }
	{\sum_{k=1}^R \mathrm{p}_k (\hatpvec^{\tau-1}) }.
\end{aligned}
\end{equation*}
Since
\begin{equation*}
\begin{aligned}
\frac{\mathrm{p}_r (\hatpvec^{\tau-1}) }
	{\sum_{k=1}^R \mathrm{p}_k (\hatpvec^{\tau-1}) }
	&= \frac{ \exp \left(-(1/2c) \lossf[r]{\tau-1}\right) }
	{\sum_{k=1}^R \exp \left(-(1/2c) \lossf[k]{\tau-1}\right)}\\
&= \univmu{r}{\tau},
\end{aligned}
\end{equation*}
therefore
\begin{equation*}
\mathrm{p}_\mathrm{avg} (\pordervec{\tau} | \hatpvec^{\tau-1})
= \sum_{r=1}^R \univmu{r}{\tau} f_\tau (\osw[r]),
\end{equation*}
meaning that the conditional pseudo-probability of the uniformly-weighted OSF can be
expressed in terms of the universal OSF weights.
By applying successive conditioning, the pseudo-probability of the uniformly-weighted
OSF becomes:
\begin{equation*}
\mathrm{p}_\mathrm{avg} (\hatpvec^\tau) =  B \prod_{t=1}^\tau \sum_{r=1}^R
\univmu{r}{t} f_t (\osw[r]).
\end{equation*}


Now we use the previous definitions to determine the pseudo-probability of the
universal OSF:
\begin{equation*}
\begin{aligned}
\mathrm{p}_\mathrm{u} (\hatpvec^\tau) &=
B \exp \left(-\tfrac{1}{2c} \lf{\tau}{\wu, \hatpvec}\right) \\
&= B \exp \left(-\tfrac{1}{2c} \sum_{t=1}^\tau (\wu^T \pordervec{t})^2 \right) \\
&= B \prod_{t=1}^\tau \exp \left(-\tfrac{1}{2c} \left(\sum_{r=1}^R \univmu{r}{t} \osw[r]^T \pordervec{t} \right)^2 \right) \\
&= B \prod_{t=1}^\tau f_t \left(\sum_{r=1}^R \univmu{r}{t} \osw[r]  \right)
\end{aligned}
\end{equation*}

Expanding the right side of the inequality in (\ref{eq:pavg_inequality}) and
applying natural logarithm to both sides yields:
\begin{equation*}
\begin{aligned}
- \ln (\mathrm{p}_\mathrm{avg} (\hatpvec^\tau)) &\leq \ln R - \min_r \ln (\mathrm{p}_r (\hatpvec^\tau) )\\
& \leq \ln R - \ln B + \tfrac{1}{2c} \min_r \lossf[r]{\tau}.
\end{aligned}
\end{equation*}

If $c \geq A \geq (\vc{w}^T \pordervec{t})^2 \forall t$, then
$f_t(\vc{w})$ will be concave. From Jensen's inequality, we have:
\begin{equation*}
f_t \left(\sum_{r=1}^R \univmu{r}{t} \osw[r]  \right)
	\geq \sum_{r=1}^R \univmu{r}{t} f_t (\osw[r]).
\end{equation*}
Therefore,
\begin{equation*}
\begin{aligned}
\mathrm{p}_\mathrm{u} (\hatpvec^\tau)
	&\geq \mathrm{p}_\mathrm{avg} (\hatpvec^\tau),\\
-\ln (\mathrm{p}_\mathrm{u} (\hatpvec^\tau))
	&\leq  \tfrac{1}{2c} \min_r \lossf[r]{\tau} + \ln(R/B),\\
\tfrac{1}{2c} \lf{\tau}{\wu, \hatpvec}
	&\leq \tfrac{1}{2c} \min_r \lossf[r]{\tau} + \ln R,\\
\tfrac{1}{\tau}\lf{\tau}{\wu, \hatpvec}
	&\leq \tfrac{1}{\tau} \min_r \lossf[r]{\tau} + \tfrac{2c}{\tau}\ln R.
\end{aligned}
\end{equation*}
Setting the hyperparamater $c$ to be equal to $A$, which is the upper bound of
the OSF output, yields Theorem \ref{thm:regret}.

Note that this proof makes only two assumptions: that the OSF is an unbiased
estimator, i.e. $\mathrm{E}\{\vc{w}\pordervec{t}\} = \psd$, and that the output of
the OSF is bounded $(\vc{w}^T \pordervec{t})^2 \leq A \forall t$.
A bounded periodogram follows from the bounded amplitude of recorded data $x[l]$.
Neither the competing weight vectors need to be in the forms Section
\ref{sec:generalized_osf} describes, nor the parent distribution of the OS
vectors $\pordervec{t}$ need to be exponential.
The same optimality proof applies to blended OSF for other probability
distributions, such as those listed in \cite{bovik_generalization_1983}.

\bibliographystyle{plain}
\bibliography{references.bib}

\begin{thebibliography}{10}

\bibitem{abraham_underwater_2019}
Douglas~A. Abraham.
\newblock {\em Underwater {{Acoustic Signal Processing}}: {{Modeling}},
  {{Detection}}, and {{Estimation}}}.
\newblock Modern {{Acoustics}} and {{Signal Processing}}. Springer
  International Publishing, Cham, Switzerland, 2019.

\bibitem{allen_findchirp_2012}
Bruce Allen, Warren~G. Anderson, Patrick~R. Brady, Duncan~A. Brown, and Jolien
  D.~E. Creighton.
\newblock {{FINDCHIRP}}: {{An}} algorithm for detection of gravitational waves
  from inspiraling compact binaries.
\newblock {\em Phys. Rev. D}, 85(12):122006, June 2012.

\bibitem{anchieta_robust_2023}
David~C. Anchieta and John Buck.
\newblock Robust power spectral density estimation via a performance-weighted
  blend of order statistics.
\newblock {\em The Journal of the Acoustical Society of America},
  154(4\_supplement):A209--A209, October 2023.

\bibitem{anchieta_robust_2024}
David~Campos Anchieta and John~R. Buck.
\newblock Robust {{Power Spectral Density Estimation With}} a {{Truncated
  Linear Order Statistics Filter}}.
\newblock {\em IEEE J. Oceanic Eng.}, pages 1--6, 2024.

\bibitem{beaton_fitting_1974}
Albert~E. Beaton and John~W. Tukey.
\newblock The {{Fitting}} of {{Power Series}}, {{Meaning Polynomials}},
  {{Illustrated}} on {{Band-Spectroscopic Data}}.
\newblock {\em Technometrics}, 16(2):147--185, May 1974.

\bibitem{bovik_generalization_1983}
A.~Bovik, T.~Huang, and D.~Munson.
\newblock A generalization of median filtering using linear combinations of
  order statistics.
\newblock {\em IEEE Trans. Acoust., Speech, Signal Process.}, 31(6):1342--1350,
  December 1983.

\bibitem{buck_performanceweighted_2018}
John~R. Buck and Andrew~C. Singer.
\newblock A {{Performance-Weighted Blended Dominant Mode Rejection
  Beamformer}}.
\newblock In {\em 2018 {{IEEE}} 10th {{Sensor Array}} and {{Multichannel Signal
  Processing Workshop}} ({{SAM}})}, pages 124--128, Sheffield, July 2018. IEEE.

\bibitem{david_order_2003}
H.~A. David and H.~N. Nagaraja.
\newblock {\em Order Statistics}.
\newblock John Wiley, Hoboken, N.J, 3rd ed edition, 2003.

\bibitem{erdim_covariance_2022}
Savas Erdim, John~R. Buck, and C.J. Berg.
\newblock Covariance {{Matrix Tapered Beamformer That Is Universal Over Notch
  Width}}.
\newblock In {\em 2022 56th {{Asilomar Conference}} on {{Signals}},
  {{Systems}}, and {{Computers}}}, pages 1413--1418, Pacific Grove, CA, USA,
  October 2022. IEEE.

\bibitem{erdim_doubly_2022}
Savas Erdim, John~R. Buck, Christopher Gravelle, C.J. Berg, and Isaiah Lacombe.
\newblock Doubly {{Adaptive Covariance Matrix Taper Universal Beamformer}}.
\newblock In {\em {{OCEANS}} 2022, {{Hampton Roads}}}, pages 1--6, Hampton
  Roads, VA, USA, October 2022. IEEE.

\bibitem{gandhi_analysis_1988}
P.~P. Gandhi and S.~A. Kassam.
\newblock Analysis of {{CFAR}} processors in nonhomogeneous background.
\newblock {\em IEEE Trans. Aerosp. Electron. Syst.}, 24(4):427--445, July 1988.

\bibitem{kay_modern_1988}
Steven~M. Kay.
\newblock {\em Modern Spectral Estimation: Theory and Application}.
\newblock Prentice-{{Hall}} Signal Processing Series. Prentice Hall, Upper
  Saddle River, N.J, 1988.

\bibitem{kay_fundamentals_1993}
Steven~M. Kay.
\newblock {\em Fundamentals of Statistical Signal Processing}.
\newblock Prentice {{Hall}} Signal Processing Series. Prentice-Hall PTR,
  Englewood Cliffs, N.J, 1993.

\bibitem{lloyd_leastsquares_1952}
E.~H. Lloyd.
\newblock Least-{{Squares Estimation}} of {{Location}} and {{Scale Parameters
  Using Order Statistics}}.
\newblock {\em Biometrika}, 39(1/2):88, April 1952.

\bibitem{ma_passive_2005}
Barry~B. Ma and Jeffrey~A. Nystuen.
\newblock Passive {{Acoustic Detection}} and {{Measurement}} of {{Rainfall}} at
  {{Sea}}.
\newblock {\em Journal of Atmospheric and Oceanic Technology},
  22(8):1225--1248, August 2005.

\bibitem{mallary_listening_2023}
C.~Mallary, C.~J. Berg, J.~R. Buck, and A.~Tandon.
\newblock Listening for rain: {{Principal}} component analysis and linear
  discriminant analysis for broadband acoustic rainfall detection.
\newblock {\em The Journal of the Acoustical Society of America},
  154(1):556--570, July 2023.

\bibitem{rabiner_applications_1975}
L.~Rabiner, M.~Sambur, and C.~Schmidt.
\newblock Applications of a nonlinear smoothing algorithm to speech processing.
\newblock {\em IEEE Trans. Acoust., Speech, Signal Process.}, 23(6):552--557,
  December 1975.

\bibitem{schwock_statistical_2021}
Felix Schwock and Shima Abadi.
\newblock Statistical {{Properties}} of a {{Modified Welch Method That Uses
  Sample Percentiles}}.
\newblock In {\em {{ICASSP}} 2021 - 2021 {{IEEE International Conference}} on
  {{Acoustics}}, {{Speech}} and {{Signal Processing}} ({{ICASSP}})}, pages
  5165--5169, Toronto, ON, Canada, June 2021. IEEE.

\bibitem{singer_universal_1999}
A.C. Singer and M.~Feder.
\newblock Universal linear prediction by model order weighting.
\newblock {\em IEEE Trans. Signal Process.}, 47(10):2685--2699, October 1999.

\bibitem{thomson_wt4_1977}
D.~J. Thomson.
\newblock {{WT4 Millimeter Waveguide System}}: {{Spectrum Estimation
  Techniques}} for {{Characterization}} and {{Development}} of {{WT4
  Waveguide-II}}.
\newblock {\em Bell System Technical Journal}, 56(10):1983--2005, December
  1977.

\bibitem{thomson_spectrum_1977}
David~J. Thomson.
\newblock Spectrum {{Estimation Techniques}} for {{Characterization}} and
  {{Development}} of {{WT4 Waveguide-I}}.
\newblock {\em Bell System Technical Journal}, 56(9):1769--1815, November 1977.

\bibitem{trucco_compounding_2022}
Andrea Trucco, Annalisa Barla, Roberto Bozzano, Emanuele Fava, Sara Pensieri,
  Alessandro Verri, and David Solarna.
\newblock Compounding {{Approaches}} for {{Wind Prediction From Underwater
  Noise}} by {{Supervised Learning}}.
\newblock {\em IEEE J. Oceanic Eng.}, 47(4):1172--1187, October 2022.

\bibitem{trucco_introducing_2023}
Andrea Trucco, Annalisa Barla, Roberto Bozzano, Sara Pensieri, Alessandro
  Verri, and David Solarna.
\newblock Introducing {{Temporal Correlation}} in {{Rainfall}} and {{Wind
  Prediction From Underwater Noise}}.
\newblock {\em IEEE J. Oceanic Eng.}, 48(2):349--364, April 2023.

\bibitem{tucker_performance_2025}
Jeff Tucker, Kathleen~E. Wage, John~R. Buck, and Lora~J. Van~Uffelen.
\newblock Performance weighted blended spectrogram.
\newblock {\em The Journal of the Acoustical Society of America},
  157(3):2106--2116, March 2025.

\bibitem{tucker_performance_2019}
Jeff~B. Tucker and Kathleen~E. Wage.
\newblock Performance {{Weighted Blending}} of {{Multiplicative}} and {{Min
  Processors}}.
\newblock In {\em 5th {{International Conference}} on the {{Effects}} of
  {{Noise}} on {{Aquatic Life}}}, page 055004, Den Haag, The Netherlands, 2019.

\bibitem{welch_use_1967}
Peter~D. Welch.
\newblock The use of fast {{Fourier}} transform for the estimation of power
  spectra: {{A}} method based on time averaging over short, modified
  periodograms.
\newblock {\em IEEE Trans. Audio Electroacoust.}, 15(2):70--73, June 1967.

\end{thebibliography}
\end{document}